\documentclass[12pt]{article}
\usepackage{amsmath}
\usepackage{amsthm}
\usepackage{graphicx}
\usepackage{epsfig}
\usepackage{amsfonts}
\usepackage{subfigure}
\usepackage{enumitem}
\usepackage{setspace}

\theoremstyle{plain}
\newtheorem{theorem}{Theorem} 
\newtheorem{lemma}{Lemma}
\newtheorem{proposition}{Proposition}
\newtheorem{corollary}{Corollary}

\newtheorem{claim}{Claim}
\newtheorem{fact}{Fact}

\theoremstyle{definition}
\newtheorem{definition}{Definition}
\newtheorem{problem}{Problem}

\newtheorem{algorithm}{Algorithm}

\theoremstyle{remark}
\newtheorem{remark}{Remark}

\DeclareMathOperator*{\argmin}{arg\,min}
\newcommand{\Comment}[1]{}
\newcommand{\cX}{{\cal X}}
\newcommand{\cE}{{\cal E}}

\newcommand{\cC}{{\cal C}}

\newcommand{\cV}{{\cal V}}

\newcommand{\cY}{{\cal Y}}

\newcommand{\cN}{{\cal N}}

\newcommand{\bb}{{\textbf b}}
\newcommand{\bc}{{\textbf c}}
\newcommand{\ba}{{\textbf a}}
\newcommand{\bp}{{\textbf p}}
\newcommand{\bx}{{\textbf x}}
\newcommand{\bs}{\mathbf{s}}
\newcommand{\by}{\mathbf{y}}

\newcommand{\hxyt}{\text{H}(X|Y^t)}

\newcommand{\E}{\mathrm{E}}
\newcommand{\prob}{\mathrm{Pr}}

\newcommand{\h[1]}{\ensuremath{ \text{H} \left( #1 \right) }}

\newcommand{\defined}{\stackrel{\triangle}{=}}

\begin{document}
% paper title
\title{Efficient Joint Network-Source Coding for Multiple Terminals with Side Information\thanks{The material in this
paper was presented in part at the Information Theory and Applications Workshop, San Diego, USA, 2011 and the IEEE International Symposium on Information Theory, Saint Petersburg, Russia, 2011.}}
% author names and affiliations
% use a multiple column layout for up to three different
% affiliations
\author{Chen Avin\thanks{Authors are with the Department of Communication Systems Engineering, Ben-Gurion University of the Negev, Israel. Email: \{avin,borokhom,coasaf,zvilo\}@bgu.ac.il}, Michael Borokhovich$^\dagger$, Asaf Cohen$^\dagger$ and Zvi Lotker$^\dagger$}
% make the title area

\maketitle
\date
\begin{abstract}
Consider the problem of source coding in networks with multiple receiving terminals, each having access to some kind of side information. In this case, standard coding techniques are either prohibitively complex to decode, or require network-source coding separation, resulting in sub-optimal transmission rates. To alleviate this problem, we offer a joint network-source coding scheme based on matrix sparsification at the code design phase, which allows the terminals to use an efficient decoding procedure (syndrome decoding using LDPC), despite the network coding throughout the network. Via a novel relation between matrix sparsification and rate-distortion theory, we give lower and upper bounds on the best achievable sparsification performance. These bounds allow us to analyze our scheme, and, in particular, show that in the limit where all receivers have comparable side information (in terms of conditional entropy), or, equivalently, have weak side information, a vanishing density can be achieved. As a result, efficient decoding is possible at all terminals simultaneously. Simulation results motivate the use of this scheme at non-limiting rates as well.
\end{abstract}
%%%%%%%%%%%%%%%%%%%%%%%%%%%%%%%%%%%%%%%%%%%%%%%%%%%%%%%%%%%%%%%%%%%%%%%%
%%%%%%%%%%%%%%%%%%%%%%%%%%%%%%%%%%%%%%%%%%%%%%%%%%%%%%%%%%%%%%%%%%%%%%%%
\section{Introduction}
In this work, we consider the problem of efficient distributed source coding for large networks with multiple receiving terminals. Assume that information source $X$ is available at some source node $s$ in a network with noiseless links. The source is to be distributed (with arbitrarily small error probability) to \emph{several} receiving nodes in the network, $T=\{t_1, \ldots t_K\}$, where each node $t\in T$ has some side information $Y^t$ available. We assume that for each $t$, $(X,Y^t)$ have some known joint distribution. An example is given in Figure \ref{fig:si}.
% the model
\def\FIGA{
\begin{figure}
\centering
\includegraphics[scale=0.7]{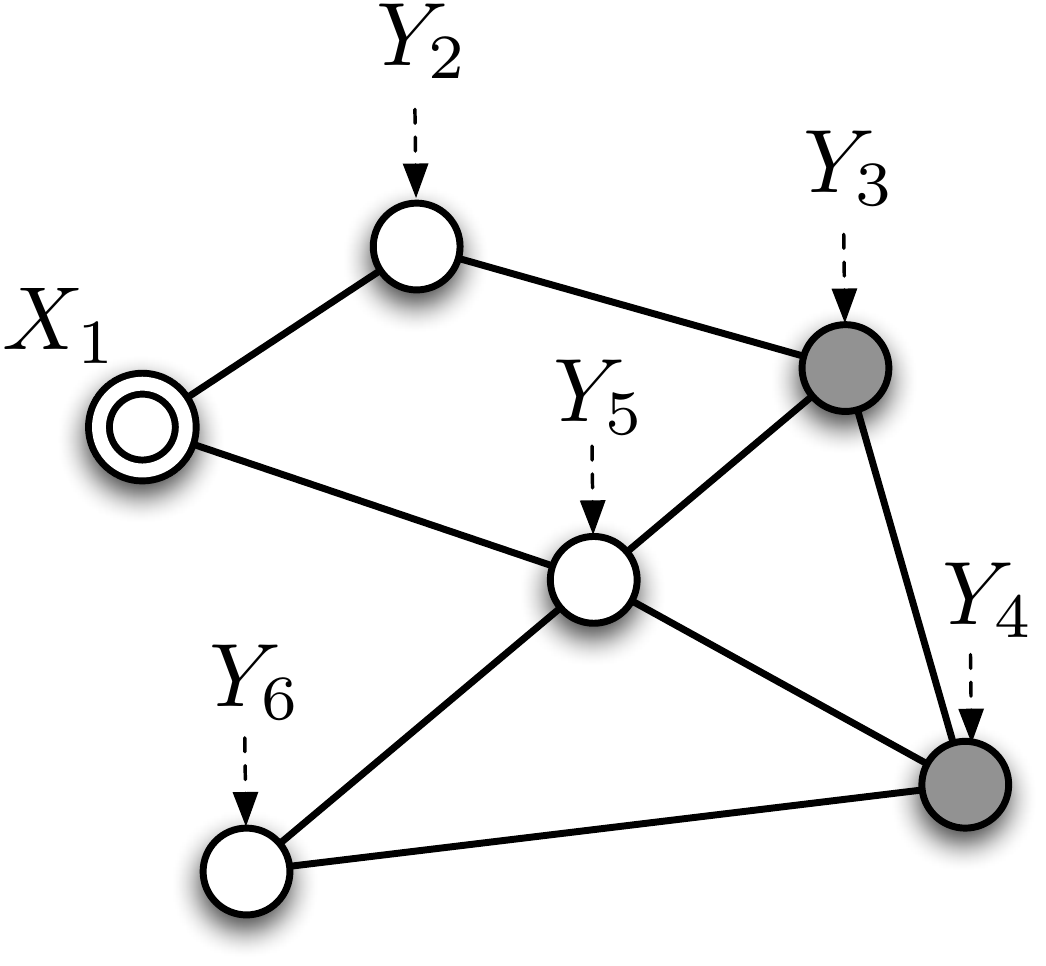}
\caption{A network with side information. $X_1$ (the source) is demanded at Terminals $3$ and $4$. Each node has a side information message $Y_i$, correlated with $X_1$.}
\label{fig:si}
\end{figure}
}

This problem arises in a multitude of networking applications, such as sensor networks, peer-to-peer, and content distribution networks. A few examples to be kept in mind can be a sensor in a network distributing a temperature measurement while each receiving node has its own altitude reading (which is highly correlated) or a video file streamed distributively in a network where receivers may have previous/noisy/low-resolution versions of that stream.  

The problem of lossless coding with side information has     
been studied extensively. We give here only a brief overview of the closely related works. In \cite{SlepianWolf73}, Slepian and Wolf considered the problem of separately encoding two correlated sources and joint decoding (Figure \ref{fig:slepian_wolf}). The asymmetric case where side information $Y$ is available at the decoder is a special case (Figure \ref{fig:asym_slepian_wolf}). In \cite{Ho04networkcoding}, Ho et al.\ considered the multicast problem with correlated sources (Figure \ref{fig:Ho05}), and completely characterized the \emph{rate region} for this problem: the set of required link capacities to support the multicast. In a way, \cite{Ho04networkcoding} can be viewed as extending the Slepian Wolf problem to arbitrary networks through \emph{network coding}. Further extensions also appeared in \cite{Barros_Servetto06} and \cite{BakshiEffros08} (Figure \ref{fig:Bakshi08}).

The canonical, three-node network of Slepian and Wolf \cite{SlepianWolf73}, as well as its extension in \cite{Ho04networkcoding} do not consider practical decoding algorithms. Such algorithms where considered in \cite{aaron2002compression} using Turbo codes, in \cite{pradhan2003distributed} based on the Wyner scheme \cite{wyner74recent} and in \cite{liveris2002compression,caire2003lossless,lan2005slepian} based on Low-Density Parity-Check Codes (LDPC), including a solution for any point on the rate region \cite{schonberg2004distributed}. In a recent study \cite{coleman2004some,coleman2005towards}, the authors use linear-programming based decoding techniques. An algebraic (Reed-Solomon-based) coding scheme was suggested in \cite{LiRamamoorthy11}. Nevertheless, a major disadvantage of applying decoding techniques based on structured linear codes to a networked environment is the need to separate network and source coding. This separation can fail for most demand structures \cite{Ramamoorthy_et_al06} (that is, will require higher rates than the necessary cut-set bounds). Hence, while the distributed source coding model was studied extensively in the information theory literature throughout the years, a huge gap still exists when trying to apply the \emph{efficient decoding schemes to large networks with many terminals}, where network coding is essential to achieve maximal throughput.

First introduced by Ahlswede et al.\ in \cite{Ahlswede_et_al_00}, network coding deals with various coding operations that can be performed at \emph{intermediate nodes} in the network in order to achieve certain rate goals. 
For linear network coding \cite{Li_et_al03,Ho_et_al06}, decoding sums up to solving a set of linear equations. When all terminals can solve and reconstruct the original data, the application layer above the network code can choose the \emph{source coding scheme} independently, as first the network code is decoded, and then the source code. This is a \emph{separation based} coding scheme. However, in general, separation schemes fail to achieve the maximal throughput in the network (e.g., a rank equal to the max-flow at each node \cite{Li_et_al03}), and an extra rate is required if each terminal wishes to decode its own subset of bits. In this case we say that \emph{separation fails} \cite{Ramamoorthy_et_al06}. The solution is thus one of the following: use extra rate (compared to the max-flow bound), or use \emph{joint network-source code}, where the network and source codes are matched and are designed to be decoded together. To date, joint network-source coding schemes are prohibitively complex, while efficient decoding schemes require a structured source code, which is shattered by the lack of separation. A midpoint approach was taken in \cite{lee2007minimum} using minimum-cost optimization, giving a trade-off between joint and separate network-source coding. In \cite{Wu09practical}, a different approach was taken, yet still sub-optimal in terms of the required rates. Alternatively, in \cite{maierbacher2009practical,Cruz11} the authors suggest a sum-product-based algorithm, nevertheless, in this case, efficient decoding requires a network with bounded degree nodes. Hence, our goal in this work is to design efficient coding schemes which are both rate-optimal and applicable to general network topology. 
\paragraph{Main Contribution:}
We formally define the setting for joint network-source coding in a network with side information at multiple terminals. We design a joint coding scheme that, under certain constraints, both achieves the lower bounds on the link capacities and facilitates efficient decoding at the terminals, by inducing, at all terminals simultaneously, a source code which is based on a low-density parity-check matrix, and hence can be decoded efficiently. This joint network-source code is achieved by a sparsification procedure of the source doding matrix and the network transfer matrix together, at the code-design stage. Our analysis of the sparsification performance is based on a novel connection between matrix sparsification and rate-distortion theory, a connection which besides being interesting on its own, allows us to both analyze the best possible sparsification performance and give a randomized algorithm to approximate it. Numerical results also illustrate the efficiency of the derived codes and algorithms. Furthermore, this connection gives a novel randomized algorithm for coding with a fidelity criterion. 

The rest of the paper is organized as follows. Section \ref{prelim} gives the required notation, related results and describes the model of distributed source coding with several receivers, including a discussion of the main difficulties. Section \ref{main} gives our main results. Section \ref{sec. ramdomized} describes the coding algorithm. Section \ref{sec. numerical} gives the numerical results and Section \ref{conc} concludes the paper.   
%%%%%%%%%%%
%%%%%%%%%%%%%%%%%%%%%%%%%%%%%%%%%%%%%%%%%%%%%%%%%%%%%%%%%%%%%%%%%%%%%%%%
\section{Preliminaries}\label{prelim}
\paragraph{Network and Source Model:} 
A network is defined as a directed acyclic graph $(\cV,\cE)$, 
where $\cV$ is the set of vertices (nodes) and 
$\cE\subseteq\cV\times\cV$ is the set of edges (links). 
Associated with each edge $e \in \cE$ is a capacity $c(e)\geq 0$. 
For any node $v$, we denote the set of incoming and outgoing edges of node $v$ by $In(v)$ and $Out(v)$, respectively.

Let $\{X_i\}_{i=1}^{\infty}$ be a sequence of independent and identically distributed random variables with alphabet $\cX$. Source $\{X_i\}$ is available at node $s \in \cV$. The case of multiple sources can be dealt with similar techniques, by coding for points on the rate region. Let $T\subseteq \cV$, $|T|=K$, be a set of terminals. A terminal $t \in T$ has (side) information $\{Y^t_i\}_{i=1}^{\infty}$ available to it. For each $t$, $\{Y^t_i\}_{i=1}^{\infty}$ is a sequence of independent and identically distributed random variables with alphabet $\cY$. We assume that the pairs $(X_i,Y^t_i)$ may be dependent, with some known joint distribution $p(x,y^t)$, yet for different time indices $i$ and $j$, $(X_i,Y^t_i)$ and $(X_j,Y^t_j)$ are independent (that is, a memoryless model).
We assume that the node $s$ has no incoming edges and that each terminal node in $T$ has no outgoing edges.

For any vector of rates $(R_e)_{e\in \cE}$, a $((2^{nR_e})_{e\in \cE},n)$ code comprises the following mappings:
$g^n_e:\cX^n\mapsto\{1,\ldots,2^{nR_e}\}$, for  $e\in Out(s)$
and
$g^n_e:\Pi_{e'\in In(v)}\{1,\ldots,2^{nR_{e'}}\}  
\mapsto\{1,\ldots,2^{nR_e}\}$, for $e\in Out(v), v\ne s$. 
That is, the source codes a block of $n$ inputs and maps it to messages on its outgoing links. Each internal node may \emph{code} its inputs and, again, map them to messages on the outgoing links. Thus, for each terminal $t\in T$, $\Pi_{e\in In(t)}\{1,\ldots,2^{nR_e}\}$ is the set of inputs to that terminal. For simplicity, we denote this input by $f_t(X^n)$. $f_t(X^n)$, together with $(Y^t_1,\ldots,Y^t_n) \defined Y^{t,n}$ is the data the terminal uses for decoding.

Throughout this work, we use \emph{linear network coding}, that is, source symbols are mapped to a vector of length $w$ over some finite field $F$. An edge $e\in Out(v)$, $v\in \cV$, carries an element of the field, which is a linear combination (over $F$) of the variables on all edges $e'\in In(v)$. Note that this ``unit capacity" assumption does allow for any rational capacity by adding parallel edges and normalizing over several network uses. It is convenient to represent the input to the source as $w$ imaginary incoming edges. Under this notation, $f_e$ denotes the global coding vector for edge $e$ (its coefficients with respect to the input variables) and for some $t\in T$, $V_t$ denotes the linear space spanned by $\{f_e: e\in In(t)\}$.

For nodes $s$ and $t$, we denote by $\mathrm{maxflow}(s,t)$ the capacity of the maximal flow between $s$ and $t$. This flow is equal to the capacity of the minimal cut separating $s$ from $t$.

\begin{definition}[\cite{yeung2008information}]
A $w$-dimensional linear network code on an acyclic network is \emph{linear broadcast} if for every non-source node $t$, $\text{rank}(V_t)=\min \{w,\mathrm{maxflow}(s,t)\}$.  
\end{definition}
\begin{corollary}[\cite{yeung2008information}]\label{cor. broadcast code}
For any acyclic network and sufficiently large base field $F$, there exists a $w$-dimensional linear broadcast code.  
\end{corollary}
In words, for any dimension $w$, there exist a large enough base field $F$ and a linear network code such that for any non-source node $t$, we have $f_t(\bs) = \bs M_t$, where $\bs$ is the $w$ dimensional source vector over $F$ and $M_t$ is a $w \times |In(t)|$ matrix over $F$, whose columns are the global coding vectors $\{f_e:e\in In(t)\}$. Furthermore, $\text{rank}(M_t)=\min \{w,\mathrm{maxflow}(s,t)\}$.

As the base field $F$ should be large enough to allow for the matrices $\{M_t\}$ to satisfy Corollary \ref{cor. broadcast code}, one cannot assume that the source alphabet, $\cX$ is of the same size as $F$. From this point on, we assume that $\cX=\{0,1\}$ and that $F=\text{GF}(2^m)$ for sufficiently large $m$. In this case, $wm$ consecutive input bits are mapped to a source vector $\bs$. 
\begin{proposition}\label{prop. can get enough equations}
Let $\cN=(\cV,\cE,s,T)$ be a network with a source $s\in \cV$ and terminals $T \subseteq \cV \setminus s$. Assume $(\cV,\cE)$ is a directed acyclic graph and that $c(e)=1$ for any $e\in \cE$. Set $w = \max_{t\in T} \mathrm{maxflow}(s,t)$ and a sufficiently large $m$. Then assuming bits $b_1,\ldots,b_{mw}$ are available at the source $s$, there exists a linear network code over base field $\text{GF}(2^m)$ such that any terminal $t \in T$ receives $m\cdot \mathrm{maxflow}(s,t)$ linearly independent equations on $b_1,\ldots,b_{mw}$.
\end{proposition}
In other words, if bits $b_1,\ldots,b_{mw}$ are available at the source, a terminal $t$ with $\mathrm{maxflow}(s,t)=l \leq w$, can receive (possibly after linear operations at that terminal) a vector of $ml$ bits, $\by = (\tilde{b}_1,\ldots,\tilde{b}_{ml})$ which is the result of $(b_1,\ldots,b_{mw}) B$, for $B \in \{0,1\}^{wm \times lm}$ with rank $ml$.  
The proof of Proposition \ref{prop. can get enough equations} is rather technical and appears in Appendix \ref{app prop can get enough equations}.
%%%%%%%%%%%%%%%%%%%%%
\paragraph{Linear Codes for Structured Binning:}
Consider the source coding problem in Figure \ref{fig:slepian_wolf}. By \cite{SlepianWolf73}, the set of achievable rates is given by $R_1 \geq \h[X|Y]$, $R_2 \geq \h[Y|X]$ and $R_1+R_2 \geq \h[X,Y]$. In the particular case where the side information $Y$ is available at the decoder (Figure \ref{fig:asym_slepian_wolf}), a rate $\h[X|Y]$ is required. Now, assume $X$ is a binary symmetric random variable, related to $Y$ through a binary symmetric channel. That is, $Y=X \oplus E$, where  $E$ is a binary random variable with $\prob(E=1)=p$ and $\oplus$ denotes XOR operation. Note that in this model $\h[X|Y]=h(p)$, where $h$ is the binary entropy function $h(p)=-p\log p -(1-p)\log(1-p)$. In this case, a coding scheme based on linear codes was suggested by Wyner \cite{wyner74recent} (see \cite{zamir2002nested} also).
% slepian wolf
\def\FIGB{
\begin{figure}
\centering
\subfigure[The Slepian-Wolf network \cite{SlepianWolf73}.]{
\includegraphics[scale=0.70]{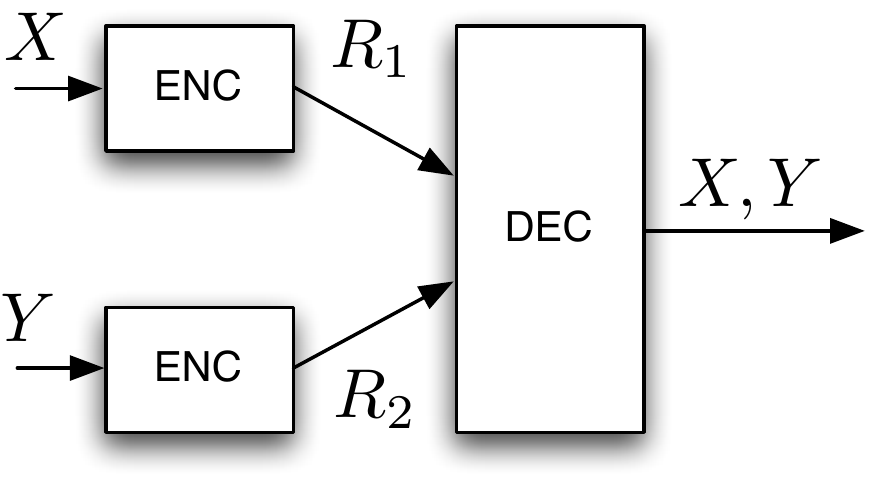}
\label{fig:slepian_wolf}
}
% asymmetric slepian wolf
\subfigure[The asymmetric Slepian-Wolf network.]{
\includegraphics[scale=0.75]{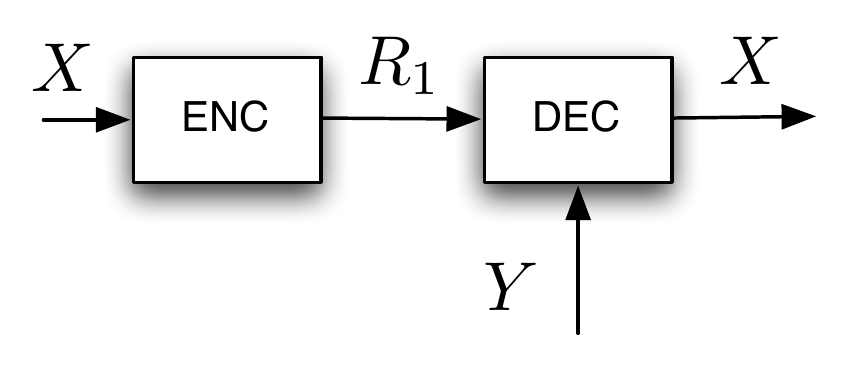}
\label{fig:asym_slepian_wolf}
}
\caption{(a) Upper encoder has the source $X$, lower one has the side information $Y$, which is, in general, correlated with $X$. We are interested in the set of rates $(R_1,R_2)$ such that both $X$ and $Y$ can be reconstructed at the decoder. (b) The encoder describes the source $X$ to a decoder which has side information $Y$ available. We are interested in the rate $R_1$ such that $X$ can be reconstructed at the decoder.}
\end{figure}
}
\def\FIGHandB{
\begin{figure}
\centering
\subfigure[A network with correlated sources \cite{Ho04networkcoding}.]{
\includegraphics[scale=0.44]{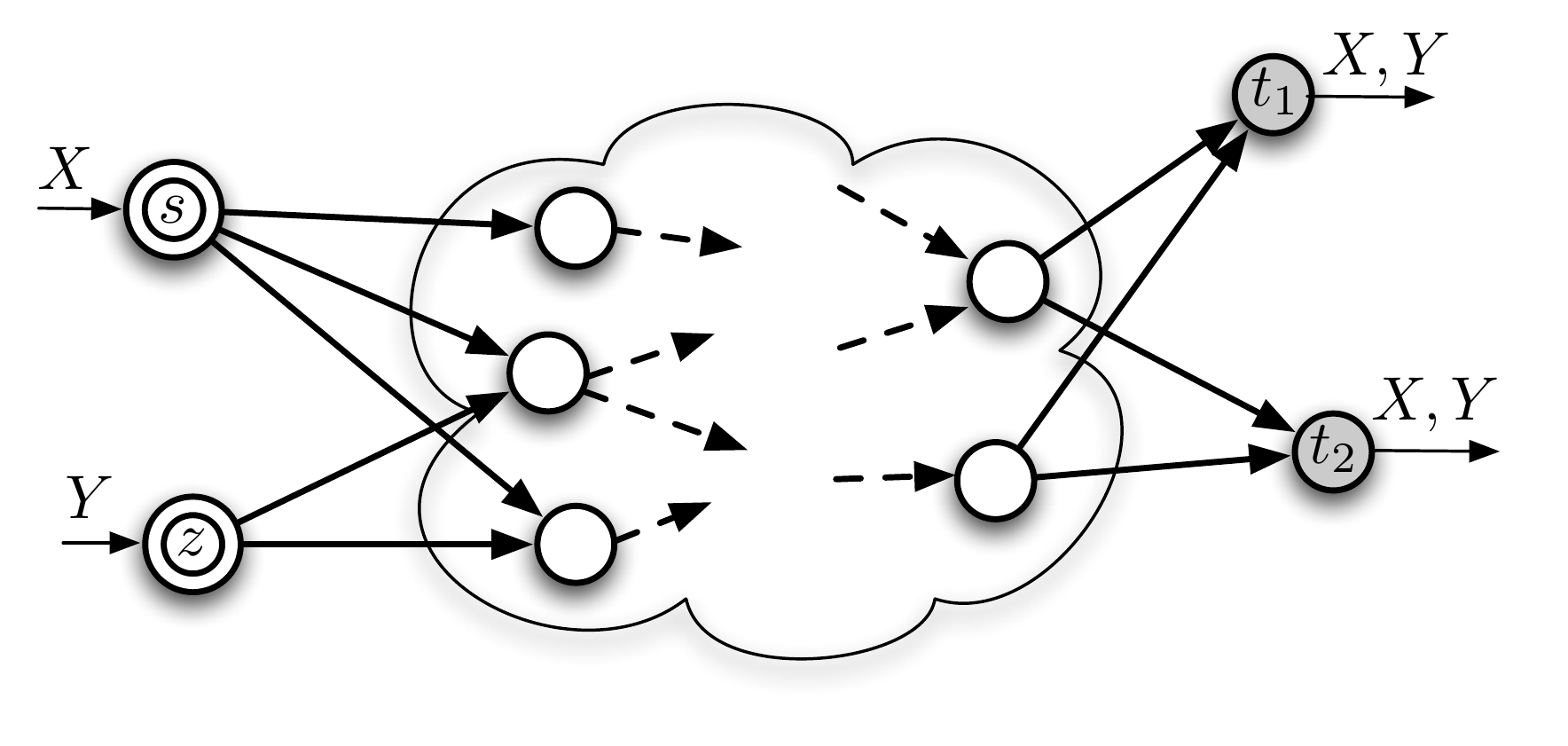}
\label{fig:Ho05}
}
% asymmetric slepian wolf
\subfigure[A network with correlated sources and side information \cite{BakshiEffros08}.]{
\includegraphics[scale=0.44]{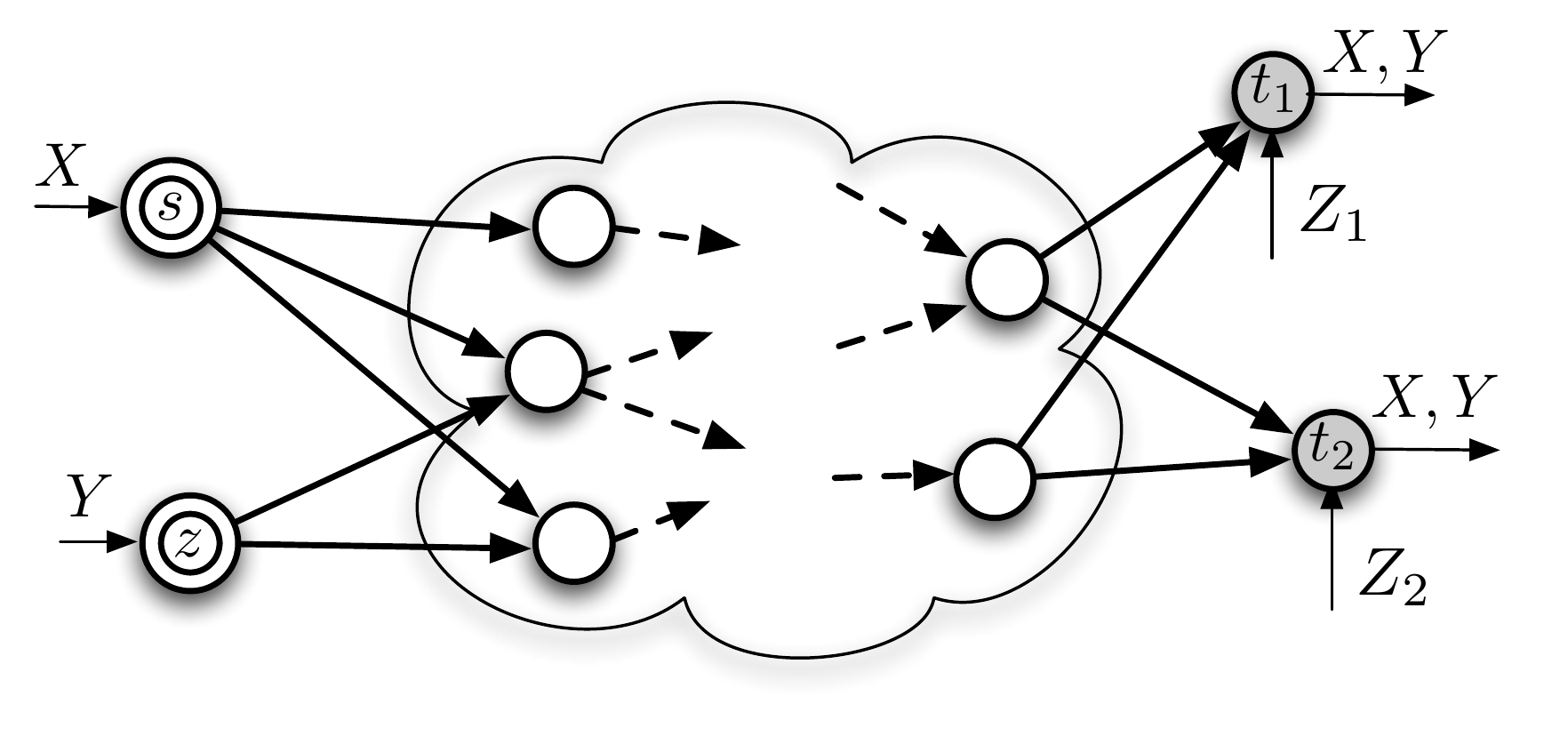}
\label{fig:Bakshi08}
}
\caption{(a) Nodes $s$ and $z$ have sources $X$ and $Y$ available. In general, $X$ and $Y$ are correlated. The sources are demanded (loselessly) at the terminals $t_1$ and $t_2$. (b) In this case, terminals $t_1$ and $t_2$ have additional side information available to them, $Z_1$ and $Z_2$, respectively.}
\end{figure}
}
\begin{fact}\label{fact wyner scheme}
Let $\epsilon>0$ be arbitrary. For $n$ sufficiently large, there exists an $n\times m$ matrix $H$ with $m < (h(p)+\epsilon)n$ and a function $f:\{0,1\}^m \mapsto \{0,1\}^n$ such that $\prob(f(e^nH)\ne e^n)\leq\epsilon$.
\end{fact}
The source vector $x^n$ is thus encoded as $x^nH$, while, at the decoder, $e^nH$ is calculated according to $e^nH = x^nH\oplus y^nH$ (as $y^n$ is available) and $x^n$ is then reconstructed using $x^n = y^n \oplus f(He^n)$. Hence, the decoding process is analogous to the decoding of a linear code with a parity check matrix $H$. For example, in \cite{liveris2002compression}, LDPCs are used. As a result, $H$ is also required to be sparse and of a certain structure.   
%%%%%%%%%%%%%%%%%%%%%%%%%%%%%%%%%%%%%%%%%%%%%%%%%%%%%%%%%%%%%%%%%%%%%%%%
\paragraph{Distributed Source Coding with Multiple Terminals:}
We now turn to our original problem. In terms of the achievable rates alone (regardless of the possibility for efficient decoding) the problem can be viewed as multicast in the presence of side information. The following can be seen as a corollary of \cite{BakshiEffros08}.
\begin{corollary}\label{cor. rates}
Let $\cN=(\cV,\cE,s,T)$ be a network with a source $s\in \cV$ and terminals $T \subseteq \cV \setminus s$. Assume $(\cV,\cE)$ is a directed acyclic graph. Let source $X$ be available at $s$ and for each $t$, side information $Y^t$ be available at $t\in T$. Then, a necessary and sufficient condition for $X$ to be reconstructed at the terminals with an arbitrarily small probability of error is $\mathrm{maxflow}(s,t) \geq \hxyt$ for all $t$.
\end{corollary} 
A possible scheme for achieving the bound in Corollary \ref{cor. rates} is \emph{random binning} at the source with \emph{linear network coding} \cite{BakshiEffros08}. In this case, it is easy to construct a source coding scheme which is oblivious to the network code, in the sense that a terminal $t$ does not care whether it receives the $k$ most significant bits of a bin index, or any $k$ linearly independent equations on the $n>k$ bits of the bin index. Another possible scheme uses linear coding and minimum entropy decoding \cite{Ho04networkcoding}. Note, however, that the decoding complexity of both schemes is exponential in the block length. Should, however, one try to use linear codes with a structure that facilitates efficient decoding, for more than two terminals the structure of the code will be destroyed by the network code. 

To see this, assume bits $x^n$ are available at the source. Denote by $w_t$ the max-flow to terminal $t$ and by $w$ the maximal max-flow among all terminals, $\max_t w_t$. By Corollary \ref{cor. rates}, there exists a (joint) network-source coding scheme for which each terminal $t$ can decode $x^n$ with an arbitrarily small (as $n \to \infty$) error probability as long as $w_t \geq  \hxyt + \epsilon$. However, due to the lack of network-source coding separation \cite{Ramamoorthy_et_al06}, even if source node $s$ creates $K$ nested codewords, were the codeword intended for $t$ is of length $n \hxyt$ (assuming $n \hxyt$ is an integer) and is a subset of the codewords for terminals with higher rates, in general, there does not exist a network code to ensure that each terminal $t$ indeed receives the subset of bits \emph{intended to it} while satisfying the min-cut bounds (this can be done, in general, for at most two terminals \cite{Erez_Feder03}). By Proposition \ref{prop. can get enough equations}, however, there exists a network code for which terminal $t'$ does receive $n H(X|Y^{t'})$ \emph{linearly independent equations} on the $n \max_t \hxyt$ transmitted bits. In other words, source node $s$ can generate a single codeword of length $n \max_t \hxyt$, using a good parity check matrix $H$ of size $n \times  n \max_t \hxyt$. Denote this codeword by $\bc = x^n H$.
$\bc$ is hence of length $n \max_t \hxyt$. Each $m$ bits in $\bc$ are mapped to a symbol in $\text{GF}(2^m)$, and a vector of $w$ symbols, $\bs$, is transmitted through the network. At terminal $t$, the received vector is $\by = \bs M_t$. We assume $wm$ divides $n \max_t \hxyt$, so $\bc$ might be transmitted using several network uses. By Proposition \ref{prop. can get enough equations}, the $n \hxyt$ \emph{bits} received at $t$ can be represented as $(b_1,\ldots,b_{n \hxyt}) = x^n H B_t$, where $B_t$ is an $n \max_{t'} H(X|Y^{t'}) \times n \hxyt$ matrix with full rank (if $|In(t)| > w_t$, one can take $w_t$ linearly independent coding vectors and discard the rest).

While $H$ is a parity-check matrix designed to facilitate efficient decoding of $e^n$ from $e^nH$, $B_t$ is a matrix \emph{defined by the network code}, and may be different for each terminal $t$. As a result, even if $H$, for example, is sparse, $HB_t$ might not be sparse at all and lack any structure that allows efficient decoding. This is the reason we say that the structure of the source code is shattered by the network code.
%%%%%%%%%%%%%%%%%%%%%%%%%%%%%%%%%%%%%%%%%%%%%%%%%%%%%%%%%%%%%%%%%%%%%%%%
\section{Efficient Decoding Using LDPC}\label{main}
The focus is thus on designing $H$ (the source code), a linear network code with matrices $B_t$, $1 \leq t \leq K$, and possibly post-processing matrices $P_t$, such that, for each terminal $t$, 
$\bar{H}_t = H B_t P_t$
is a good parity-check matrix for terminal $t$, with a structure that facilitates efficient decoding. In this section, we center our attention on $H$ and the $P_t$'s, and derive sufficient conditions to guarantee that each terminal $t$ sees, with high probability, a low-density parity-check matrix $\bar{H}_t$, assuming $B_t$ is the result of either deterministic or random linear network coding.

At the heart of the suggested scheme is a matrix sparsification algorithm and its analysis for high rates (almost square matrices). In short, for each $H B_t$, we wish to find a matrix $P_t$ such that $H B_t P_t$ is sparse. In \cite{neylon2006sparse,gottlieb2010matrix}, the authors consider randomized matrix sparsification algorithms. If $A$ is an $n \times (n-k)$ matrix, the authors seek an $(n-k) \times (n-k)$ matrix $P$ such that $A P$ is sparse. The analysis therein, however, aims at bounding the difference between the sparsity achieved by the algorithm and the best possible performance, without quantifying the best possible performance directly. In the next subsection, we show a novel connection between matrix sparsification and rate distortion theory. Through this connection, we are able to give bounds on the best possible performance directly, and, as a result, give sufficient conditions under which one can indeed find such $P$ which yields a sparse matrix $A P$. At the basis of our results are Lemmas \ref{lem lower bound} and \ref{lem achievability mat spars}, which give lower and upper bounds on the sparsification performance. Theorem \ref{theo high rate} below, our main result in this section, utilizes these lemmas to show that indeed, as long as the difference in the strength of the side information (that is, conditional entropies) available at the nodes is small, designing joint network-source codes which induce low-density parity-check matrices at the terminals is possible.
%%%
\begin{theorem}\label{theo high rate}
Let $(\cV,\cE,s,T)$ be a network with binary uniform source $X$ available at node $s$ and terminals $t \in T$ with side information $Y^t$. Assume $\mathrm{maxflow}(s,t) \geq R_t = \hxyt+\epsilon$. Then, for large enough block length $n$, it is possible to construct a joint source and network code such that each terminal can decode $X$ using a syndrome resulting from a parity check matrix $\bar{H}_t$ with normalized density at most $D(\max_{t'}R_{t'}-R_t)+O(\frac{\log n}{n})$, where $D(\cdot)$ is the distortion rate function of a binary uniform random variable under Hamming distortion. 
\end{theorem}
This result gives a vanishing density whenever the strength of the side information at the nodes is approximately the same. In particular, at the limit of high rates, a vanishing density can be achieved. A small (but still useful for efficient decoding) density can be achieved when the variation in the rates is slightly larger. Moreover, note that a naive matrix sparsification based on Gauss elimination would result only in a linear rate-sparsity trade-off, as the lower part of the matrix (of size $n(1-R) \times nR$) is arbitrary. For $R \to 1$, the improvement over Gauss elimination is arbitrarily large. To see this, note that the derivative of $D(R)$ tends to $0$ as $R\to 1$, compared to the $1/2$ rate achieved by Gauss elimination. For example, taking $R = 1-O(\frac{(\log n)^2}{n})$ results in a density of $O(\frac{\log n}{n})$. All this can be clearly seen from Figure \ref{fig:approach_rd}, which compares $D(R)$ to Gauss elimination, together with the results of the sparsification algorithm we used.

Theorem \ref{theo high rate} results from analyzing both the properties of the joint network-source code $HB_t$ seen by terminal $t$, and the best possible sparsification performance. The first part is rather technical and is summarized in Proposition \ref{prop properties} at the end of this section. We now focus on the second part and show that, at least in the randomized setting, performance guarantees for matrix sparsification can indeed be achieved via the distortion-rate function. 
\subsection{Analyzing Matrix Sparsification Via Rate-Distortion}
Matrix sparsification is formally defined as follows.
\begin{problem}\label{prob mat spar}
Given a full rank matrix $A \in \text{GF}(q)^{n\times (n-k)}$, find and invertible matrix $P \in \text{GF}(q)^{(n-k)\times (n-k)}$ which minimizes the number of non-zeros (nnz) in $AP$.
\end{problem}
In \cite{neylon2006sparse,gottlieb2010matrix} the authors prove that matrix sparsification is NP-hard, and give approximation algorithms. Note that, in our model, matrix sparsification is part of the \emph{code design}, and not the decoding process. Hence, with successful sparsification, efficient decoding is possible throughout the transmission of the (possible many) source blocks. 

The analysis in \cite{neylon2006sparse,gottlieb2010matrix} bounds the performance compared to the optimum, but does not include any guarantee on the achieved sparsity. In the context of our problem, such a guarantee is essential in order to verify that each terminal actually receives a syndrome which can be seen as created by a low-density matrix. 

First, a few definitions are required. For any $a,b \in \text{GF}(q)$ define
$
d(a,b) = \left\{ \begin{array}{ll}
0 & \textrm{if $a=b$}\\ 1 & \textrm{if $a \ne b$}.
\end{array} \right.$
With a slight abuse of notation, for any $\ba,\bb \in \text{GF}(q)^n$ define
$d(\ba,\bb) = \sum_{i=1}^{n}d(\ba(i),\bb(i))$. It is common to denote $d(\ba,\bb)$ as $\mathrm{nnz}(\ba-\bb)$.
At the heart of the matrix sparsification algorithm stands the following problem.
\begin{problem}\label{prob 1}
\textit{Min-Unsatisfy:} For any $A \in \text{GF}(q)^{n\times (n-k)}$ and $\bb \in \text{GF}(q)^n$, find the vector $\bx \in \text{GF}(q)^{n-k}$ which minimizes $d(A\bx,\bb)$.
\end{problem}
Let $A_{-j}$ denote the matrix $A$ with the $j$-th column, $\ba_j$, removed.
Under these definitions, a possible matrix sparsification algorithm has the following form \cite{neylon2006sparse}.
\begin{algorithm}\label{alg mat_spar}
For each $1\leq j \leq m$:
\begin{enumerate}[itemsep=0pt,parsep=0pt,topsep=0pt, partopsep=0pt]
\item Let $\bx = \text{Min-Unsatisfy}(A_{-j},\ba_j)$.
\item Replace the column $\ba_j$ with $\ba_i-A_{-j} \bx$.
\end{enumerate}
\end{algorithm}
The following lemma bounds the performance of \emph{any algorithm} for solving Problem \ref{prob mat spar}. 
% lemma - matrix sparsification
\begin{lemma}\label{lem lower bound} 
Let $A \in \text{GF}(q)^{n\times (n-k)}$ be a random matrix with i.i.d.\ entries. Let $P(A)\in \text{GF}(q)^{(n-k)\times (n-k)}$ be any invertible matrix, whose entries may depend on those of $A$. The expected number of non-zeros in $AP$ satisfies 
$
\frac{1}{(n-k)n}\E \left\{\mathrm{nnz}(AP)\right\} \ge D((n-k)/n)$,  
where $D(\cdot)$ is the distortion-rate function of $A(1,1)$ under Hamming distortion.
\end{lemma}
\begin{proof}
The proof is based on the converse to the rate-distortion theorem (e.g., \cite[Section 10.4]{cover2006elements}). In particular, we show that for any possible choice of $P$ (which, in general, depends on the \emph{realization} of $A$), each column of $AP$ is at most as sparse (in expectation) as the vector of differences between a column of $A$ and the closest vector to it among the \emph{best} rate-distortion code of size $q^{n-k-1}$.
  
Let $P$ be a given invertible matrix over $\text{GF}(q)$. For each column $\bp_j$ in $P$, denote by $i^*(j)$ the index of the first non-zero entry in $\bp_j$ and let $\lambda_j=1/P(i^*(j),j)$. Let $\Lambda$ denote the diagonal matrix whose diagonal entries are $\{\lambda_j\}_{j=1}^{n-k}$. Finally, denote the minimizer in Problem \ref{prob 1} by $\hat{\bx}_A(\bb)$, that is, $\hat{\bx}_A(\bb) = \argmin_\bx d(A\bx,\bb)$. We have:
\begin{multline}
\frac{1}{(n-k)n}\E\{\mathrm{nnz}(AP)\} = \frac{1}{(n-k)n}\E\{\mathrm{nnz}(AP\Lambda)\}
\\
=\frac{1}{(n-k)n}\E\left\{\sum_{j=1}^{n-k}\mathrm{nnz}(A\bp_j \lambda_j)\right\}
=\frac{1}{(n-k)n}\sum_{j=1}^{n-k}\E\left\{d\Big(\ba_{i^*(j)},-\sum_{l = i^*(j)+1}^{n-k}\ba_l P(l,j)\lambda_j\Big)\right\}
\\
\ge\frac{1}{(n-k)n}\sum_{j=1}^{n-k}\E\left\{d\Big(\ba_{i^*(j)},-A_{-j}\hat{\bx}_{A_{-j}}(\ba_{i^*(j)})\Big)\right\}. 
\end{multline}
Now, simply consider $\ba_{i^*(j)}$ as a source vector of length $n$, and the set $\{A_{-j}\bx:\bx \in \text{GF}(q)^{n-k-1}\}$ as a rate-distortion code of rate $(n-k-1)/n$ (that is, having $q^{n-k-1}$ codewords if $A_{-j}$ has full rank). The structure of the code is, of course, defined by the matrix $A_{-j}$. By the converse to the rate-distortion theorem, for \emph{any} subset $\{\bb_1,\ldots,\bb_{q^{nR}}\} \subseteq \text{GF}(q)^n$, where $0\leq R \leq 1$, we have
\begin{equation}\label{converse to RD}
\E\{d(\ba_{i^*(j)},\bb_{\hat{i}})\} \ge nD(R),
\end{equation}
where $\hat{i}$ is the minimizer of $d(\ba_{i^*(j)},\bb_i)$ over all $1\leq i \leq q^{nR}$. A-fortiori, \eqref{converse to RD} holds if the set $\{\bb_1,\ldots,\bb_{q^{nR}}\}$ is not chosen optimally (herein, it is chosen as the linear subspace spanned by the columns of $A_{-j}$). As a result,
\begin{eqnarray}
\frac{1}{(n-k)n}\sum_{j=1}^{n-k}\E\left\{d\Big(\ba_{i^*(j)},-A_{-j}\hat{\bx}_{A_{-j}}(\ba_{i^*(j)})\Big)\right\}&\ge& \frac{1}{(n-k)n}\sum_{j=1}^{n-k}nD\left(\frac{n-k-1}{n}\right)
\nonumber\\
&\ge& D\left(\frac{n-k}{n}\right),\nonumber
\end{eqnarray}
where the last inequality follows from the monotonicity of $D(R)$.
\end{proof}
Note that the above proof is based on the fact that a solution to Problem \ref{prob mat spar} can be seen as sequentially applying Problem \ref{prob 1} to each of the columns of $A$, with the rest of the columns as the second argument. This is also the reason that an approximation for Problem \ref{prob 1} results in an approximation for Problem \ref{prob mat spar}, and bounds on the possible performance in Problem \ref{prob 1} give bounds on the performance in Problem \ref{prob mat spar}. 

In fact, in certain cases, the bound in Lemma \ref{lem lower bound} is achievable. For binary random matrices with uniform i.i.d.\ entries, we have the following achievability result.
% upper bound (achieveability)
\begin{lemma}\label{lem achievability mat spars}
Let $A \in \text{GF}(2)^{n\times (n-k)}$ be a random matrix with uniform i.i.d.\ entries. For any $0 \leq D \leq 1/2$ which satisfies $\frac{n-k}{n} \ge 1-h(D)+2\frac{\log n}{n} + \frac{1}{n}$, there exists an invertible matrix $P \in \text{GF}(2)^{(n-k)\times (n-k)}$ such that with probability at least $1-(n-k)2^{-n}$, for all $1 \leq j \leq n-k$, we have $\mathrm{nnz}(A\bp_j) \leq n D$. In particular, $\frac{1}{n(n-k)}\mathrm{nnz}(AP) \leq D$.
\end{lemma}
\begin{proof}
At the basis of the proof stands the fact that for binary uniform random variables, \emph{linear codes} achieve the rate-distortion function \cite{cohen83nonconstructive}. 

First, we show how to construct a square matrix $P$ which achieves the required sparsity of $AP$.
For any matrix $C \in \text{GF}(2)^{n\times nR}$, $0\leq R \leq 1$, define $\cC = \{C\bx:\bx \in \text{GF}(2)^{nR}\}$. In \cite{cohen83nonconstructive}, it is shown that if $C$ is random with uniform i.i.d.\ entries, and $\bb$ is a binary vector with uniform i.i.d.\ entries, then if $R \ge 1-h(D)+2\frac{\log n}{n}$ we have
\begin{equation}\label{eq. res from Cohen}
\prob\left(\min_{\bc \in \cC}d(\bc,\bb) > nD \right) \leq 2^{-n}.
\end{equation}
Note that the codeword $\bc$ which minimizes $\min_{\bc \in \cC}d(\bc,\bb)$ is simply a linear combination of the columns of $C$, that is, $C\hat{\bx}_C(\bb)$. Hence, to construct the invertible matrix $P$, we proceed inductively as follows. We consider $\ba_1$ as the source vector ($\bb$ in \eqref{eq. res from Cohen}) and $A_{-1}$ as the linear code. Note that both are random with uniform i.i.d.\ entries. Define a matrix $S_1$ such that all its columns except the first one are equal to those of the identity matrix. For the first column, set $S_1(1,1)=1$ and the rest of the values as simply $\hat{\bx}_{A_{-1}}(\ba_1)$. $S_1$ is invertible.
To continue, construct $S_j$ as follows. Set all its columns except the $j$-th equal to those of the identity matrix. For the $j$-th column, set $S_j(j,j)=1$ and the rest of the values as simply $\hat{\bx}_{(AS_1 \cdots S_{j-1})_{-j}}((AS_1 \cdots S_{j-1})_j)$.  

We set the matrix $P$ as simply $S_1 \cdots S_{n-k}$. Clearly, $P$ is invertible. Note that the $j$-th column of $P$ is simply the $j$-th column of $S_1\cdots S_j$, that is $\bp_j=(S_1\cdots S_j)_j$. Moreover, this column is chosen to minimize the sum of $(AS_1 \cdots S_{j-1})_j$ and the linear combination of the columns of $(AS_1 \cdots S_{j-1})_{-j}$. 

We now show that for each $j$, $AS_1 \cdots S_{j-1}$ remains random with uniform i.i.d.\ entries. Each matrix $S_i$ can be represented as a multiplication of $s_i-1$ matrices, $s_i$ being the number of non-zero entries in the $i$th column of $S_i$. Denote them by $W_{S_i}^1,\ldots,W_{S_i}^{s_i-1}$. Each of these matrices differs from the identity matrix by one entry only: $W_{S_i}^l$ has $1$ at the entry corresponding to the $l$th non-zero element of the $i$th column of $S_i$ (excluding the diagonal element). For example, 
$\tiny\left( \begin{array}{ccc}
1 & 1 & 0 \\
0 & 1 & 0 \\
0 & 1 & 1 \end{array} \right)\ 
= 
\left( \begin{array}{ccc}
1 & 1 & 0 \\
0 & 1 & 0 \\
0 & 0 & 1 \end{array} \right)\ 
\cdot
\left( \begin{array}{ccc}
1 & 0 & 0 \\
0 & 1 & 0 \\
0 & 1 & 1 \end{array} \right)\ 
\normalsize.
$   
Thus, $S_1 \cdots S_{j-1}=W_{S_1}^1,\cdots,W_{S_1}^{s_1-1}\cdots W_{S_{j-1}}^1,\cdots,W_{S_{j-1}}^{s_{j-1}-1}$. Consider the multiplication $AW_{S_1}^1$. This operation simply replaces one column in $A$ with its XOR with another column. Since for two independent uniform bits $X$ and $Y$, $X$ and $X \oplus Y$ are also independent uniform bits, $AW_{S_1}^1$ remains a matrix with binary uniform i.i.d.\ entries. By induction, $AS_1 \cdots S_{j-1}$ is also uniform with i.i.d.\ entries.

Utilizing the above, to compute the probability that some column of $AP$ has a density higher than $nD$, we have
\begin{multline}
\prob\Big(\cup_{1 \leq j \leq n-k} \mathrm{nnz}\big(A\bp_j\big) > n D\Big) \leq \sum_{1 \leq j \leq n-k}\prob\Big(\mathrm{nnz}\big(A\bp_j\big) > n D\Big)
\\
 \quad\quad\quad\quad\quad\quad= \sum_{1 \leq j \leq n-k}\prob\Big(\mathrm{nnz}\big(A(S_1\cdots S_j)_j\big) > n D\Big)
\\
 = \sum_{1 \leq j \leq n-k}\prob\Big(\mathrm{nnz}\Big((AS_1 \cdots S_{j-1})_{-j}\cdot
  \hat{\bx}_{(AS_1 \cdots S_{j-1})_{-j}}\big((AS_1 \cdots S_{j-1})_j\big)\Big) > n D\Big)
\\
 = (n-k)\prob\Big(\mathrm{nnz}\big(A_{-1}\hat{\bx}_{A_{-1}}(\ba_1)\big) > n D\Big)
 \leq (n-k)2^{-n},
\end{multline} 
where the last inequality applies if $\frac{n-k-1}{n} \ge 1-h(D)+2\frac{\log n}{n}$. 
\end{proof}
\begin{remark}
Note that if each column in $P$ is chosen \emph{separately} of the others, in a way that the $n-k-1$ entries of $\bp_j$ minimize the distance between $\ba_j$ and some linear combination of the columns in $A_{-j}$, then $P$ is not necessarily invertible. However, it is not hard to see that all its off-diagonal entries are uniform i.i.d., while all the diagonal entries are deterministic and equal to $1$. To compute the probability that $P$ is invertible, note that the first row of $P$ has $2^{n-k-1}$ possibilities. The $i$-th row of $P$ cannot be a linear combination of the first $i-1$ rows (there are $2^{i-1}$ such combinations), yet, since $P(i,i)=1$, half of the combinations are not counted. As a result, the probability that such a $P$ is invertible is
\[
\frac{(2^{n-k-1})(2^{n-k-1}-2^1/2)(2^{n-k-1}-2^2/2)\cdots(2^{n-k-1}-2^{n-k-1}/2)}{2^{(n-k)^2-(n-k)}}=\Pi_{i=1}^{n-k-1}(1-2^{-i}) \approx 0.29.
\] 
\end{remark}
%%%%%%%%%%%%%%
\subsection{Proof of Theorem \ref{theo high rate}}\label{app. sketch of main result}
At the basis of the proof is the following theme: we identify the properties of the transfer matrix $HB_t$ seen by a terminal $t$. As $B_t$ is defined by the network code, hence constrained by the network topology and, in general, cannot be optimized to our needs, we seek a proper $H$ such that, with high probability, either $HB_t$ is sparse or it can be sparsified. Since the constraints under which $HB_t$ may be sparse are not necessarily met (depending on $B_t$), we take the freedom of choosing $H$ such that at least the result in Lemma \ref{lem achievability mat spars} can be used, and $HB_t$ can be sparsified.

The proof is based on the following proposition, which gives a qualitative description of the properties of the transfer matrices $HB_t$ seen by each terminal $t$. The results herein are stronger than those required to use Lemma \ref{lem achievability mat spars}, but we include them for completeness. A proof is given in Appendix \ref{app prop properties}.
\begin{proposition}\label{prop properties}
Assume $H$ is an $n \times (n-k)$ random binary matrix with i.i.d.\ components, where $n-k = \Theta(n)$ and $\prob(H(1,1)=1)=\lambda/n$ for some $\lambda \ge 1$. Then,
\begin{enumerate}\itemsep = -0.2cm
\item If $B_t$ is full rank then each column of $HB_t$ is composed of independent bits (that is, column-wise independence).
\item If $B_t$ is full rank and $H$ is chosen at random with \emph{uniform} entries ($\lambda = n/2$), then $HB_t$ has uniform i.i.d.\ entries.
\item If each column in $B_t$ contains a linear (in $n$) number of non-zero entries, e.g., $c_j n$ then
\\ \mbox{$
\lim_{n\to\infty} \prob((HB_t)(i,j) = 0) = \frac{1}{2}\left(1+e^{-2c_j\lambda}\right)$},
where $(HB_t)(i,j)$ denotes the $(i,j)$th entry of $HB_t$. If, in addition, $\lambda = w(1)$, we have
\\ \mbox{$
\lim_{n\to\infty} \prob((HB_t)(i,j) = 0) = \frac{1}{2}$}.
\item If each column in $B_t$ contains a sub-linear number of non-zero entries then
\\ \mbox{$
\lim_{n\to\infty} \prob((HB_t)(i,j) = 1) = 0$}. 
\item Let $\bb_{j}$ and $\bb_{j'}$ be two columns of $B_t$. If $d(\bb_{j},\bb_{j'}) = \Theta(n)$, then by setting $\lambda = \omega(n)$ the random variables $(HB_t)(i,j)$ and $(HB_t)(i,j')$ are independent for each $i$. Hence, $HB_t$ has independent columns. If, however, $d(\bb_{j},\bb_{j'}) = o(n)$, then by choosing $\lambda$ as some positive constant, we have 
$\lim_{n\to \infty}\prob\left( (HB_t)(i,j) \ne (HB_t)(i,j)\right) \to 0$. 
\end{enumerate}
\end{proposition}  
%We are now able to give a proof for Theorem \ref{theo high rate}. 
\begin{proof}[Proof of Theorem \ref{theo high rate}]
First, by Proposition \ref{prop properties}, for each terminal $t$, the transfer matrix $HB_t$ can either be made sparse using a proper choice of $H$, or it can be made completely i.i.d.\ with binary uniform entries. This is since $\lambda(n)$ is a parameter of the source code, and can be chosen by the code designer.

Then, at the ``worst case" of $HB_t$ being a uniform i.i.d.\ matrix, we use Algorithm \ref{alg mat_spar}, which results, at the limit of infinitely many repetitions, in a sparsity ratio of $2D(R)$ (Lemma \ref{lem achievability mat spars} and \cite{neylon2006sparse}). Note that by \cite[Theorem 8]{gottlieb2010matrix}, there are unsparsifiable matrices, that is, matrices which cannot be made sparser that a linear rate-sparsity curve. In a sense, our results show that under a proper choice of $H$, we can guarantee that the transfer matrices seen by the terminals \emph{can} be sparsified below the linear curve - up to the distortion-rate curve.

Remember that for $q=2$ and $\bb(0)$ being Bernoulli($p$), we have 
 \[
R(D) = \left\{ \begin{array}{ll}
h(p)-h(D) & \textrm{if $0\leq D \leq \min\{p,1-p\}$}\\ 0 & \textrm{if $D >\min\{p,1-p\}$}
\end{array} \right.
\]
and $D(R)$ is the inverse of $R(D)$. Thus, for example, from analyzing $D(R)$ for $R \to 1$, we have that for $R = 1-\frac{(\log n)^2}{n}$, $D(R)=\frac{\log n}{n}$.
\end{proof}%%%%%%%%%%%%%%
\section{A Randomized Algorithm for Linear Rate-Distortion}\label{sec. ramdomized}
Note that through the relation between matrix sparsification and rate-distortion, we are now able to give a randomized algorithm to approximate the closest codeword in a rate-distortion problem (using linear codes) and assess its approximation factor compared to the optimum.
Denote by $C(I,:)$ the matrix consisting of rows $I$ of the matrix $C$, for some index set $I$. Consider the following algorithm, whose performance guarantee is a direct application of Lemma \ref{lem achievability mat spars} and \cite{neylon2006sparse}. 
\begin{algorithm}\label{alg randomized rd}
Input: A linear code $\cC$ of rate $R$ (defined by an $n \times nR$ matrix $C$). A vector $\bb$.
Output: A code word \bc.
\begin{enumerate}[itemsep=0pt,parsep=0pt,topsep=0pt, partopsep=0pt]
\item Randomly choose an index set $I$ of $nR$ independent rows in $C$,
\item Solve for $\bx$ in $C(I,:)\bx = \bb(I)$,
\item Return $C\bx$.
\end{enumerate}
\end{algorithm}
\begin{corollary}Assume $C$ and $\bb$ are i.i.d.\ binary symmetric. With probability at least $e^{-d}$, Algorithm \ref{alg randomized rd} returns a codeword $\bc$ whose distance from $\bb$ is at most $(\frac{nR}{d}+2)\cdot nD(R)$.
\end{corollary}
%%%%%
\section{Numerical Results}\label{sec. numerical}
Figure \ref{fig:approach_rd} gives a numerical example of how close one can get to the rate-distortion curve using Algorithm \ref{alg randomized rd}. This algorithm is at the basis of the sparsification process and, in fact, the density of the matrices we use for Figure \ref{fig:rate02} is that depicted by the approximation curves in Figure \ref{fig:approach_rd}. The sparsification algorithm used was Algorithm \ref{alg mat_spar}, with a random guess to solve the Min-Unsatisfy problem (similar to Algorithm \ref{alg randomized rd}).
% RD approximation
\def\FIGC{
\begin{figure}
\centering
\includegraphics[scale=0.45]{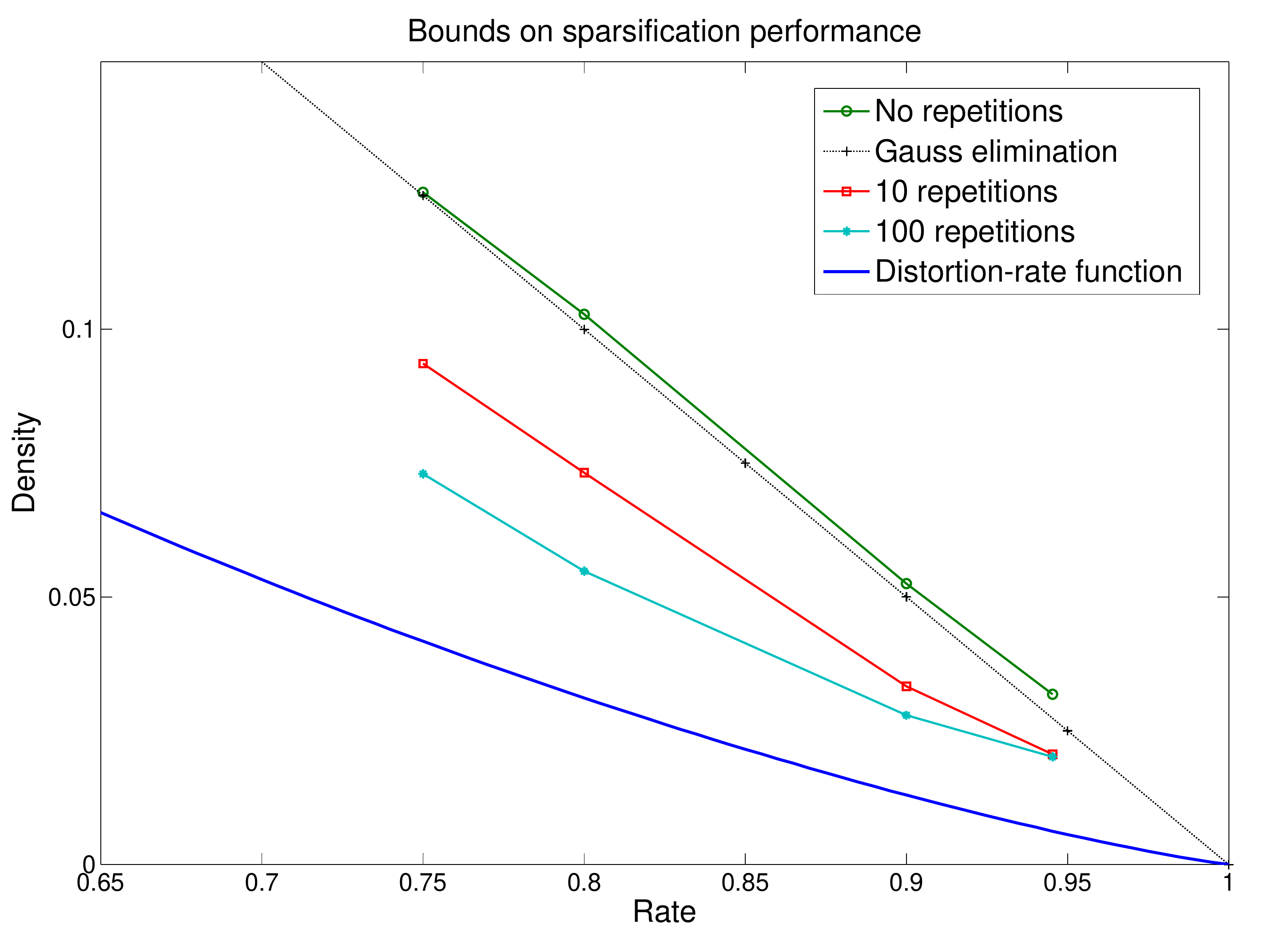}
\caption{Approaching the distortion-rate curve using Algorithm \ref{alg randomized rd}.}
\label{fig:approach_rd}
\end{figure}
}

To depict the performance of LDPC codes created from matrix sparsification of a random i.i.d.\ matrix with uniform bits, Figure \ref{fig:rate02} gives the bit error rates (BER) for decoding of such codes. Figure \ref{fig:rate02} shows also a bit error rate for two matrices created by the following \emph{structured} approach: an identity matrix is replicated $4\times 5$ times, followed by about $10^5$ permutations, where every permutation preserves the constant number of ones in each column (4 ones) and in each row (5 ones). While the results of LDPC codes created from matrix sparsification are inferior to well-optimized LDPCs, note that the codes suggested in this work do not require network-source coding separation, and hence, unlike optimized codes, can be used in networks with multiple receivers and without the (possibly very large) excess capacity above the min-cuts required when separation-based scheme is used.
% rate 02 figure
\def\FIGD{
\begin{figure}
\centering
\includegraphics[scale=0.45]{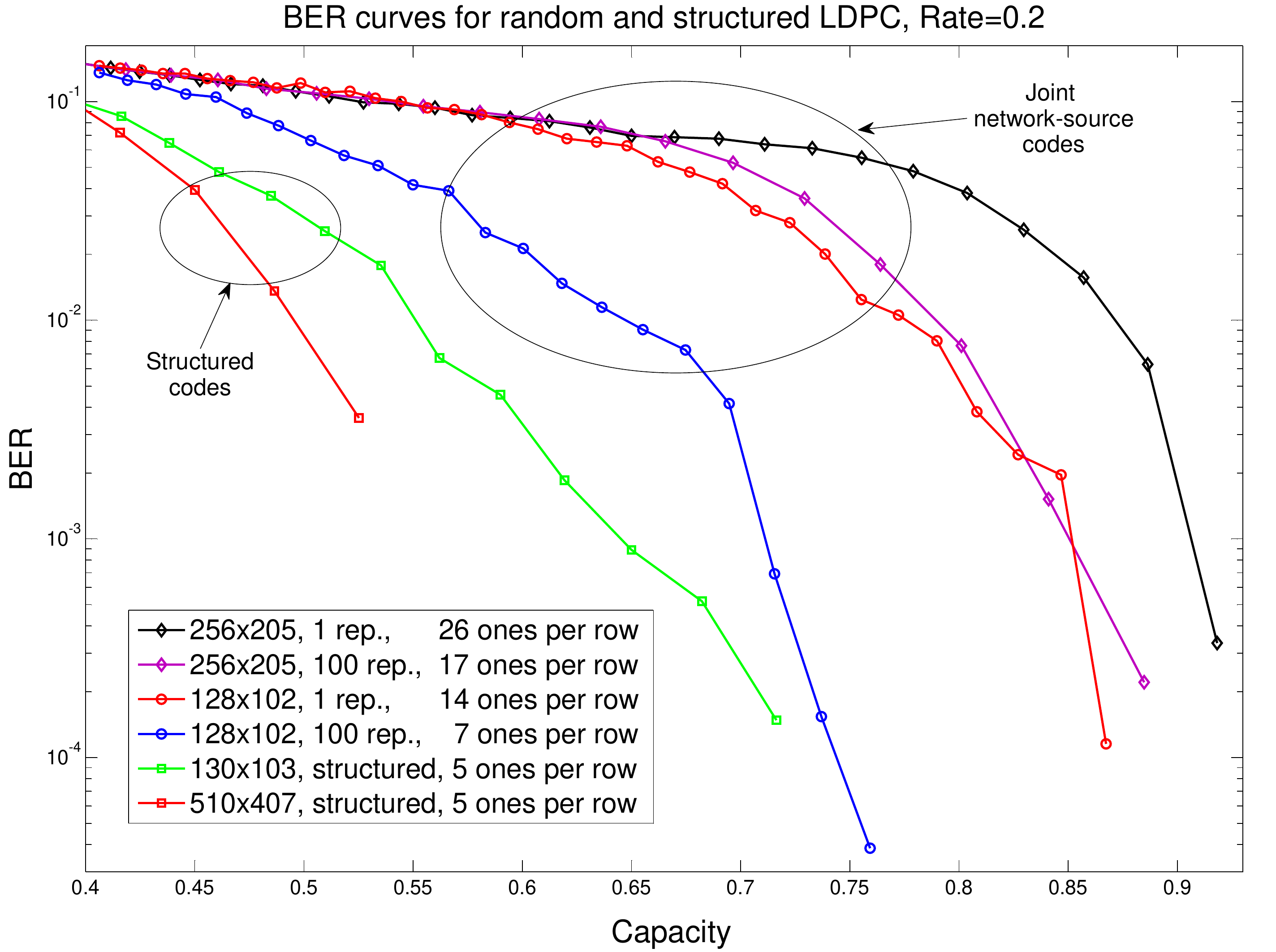}
\caption{Bit error rates for LDPCs created by applying randomized sparsification algorithm on binary symmetric i.i.d.\ matrices of size $n \times nR$, $R=0.8$ (a rate $0.2$ channel code), and for LDPCs created by a structured construction.}
\label{fig:rate02}
\end{figure}
}
Sparse matrices can also be achieved for deterministic network coding matrices, resulting from the polynomial time algorithm of \cite{jaggi2005polynomial}. For example, sparsification results for a graph of $70$ nodes, with coefficients over $\text{GF}(2^4)$, can be found in Figure \ref{fig:sparsification table}. Note that the given percentage of non-zero elements is over $\text{GF}(2^4)$, and not over $\text{GF}(2)$. 
\def\FIGE{
\begin{figure}
\centering
\includegraphics[scale=0.85]{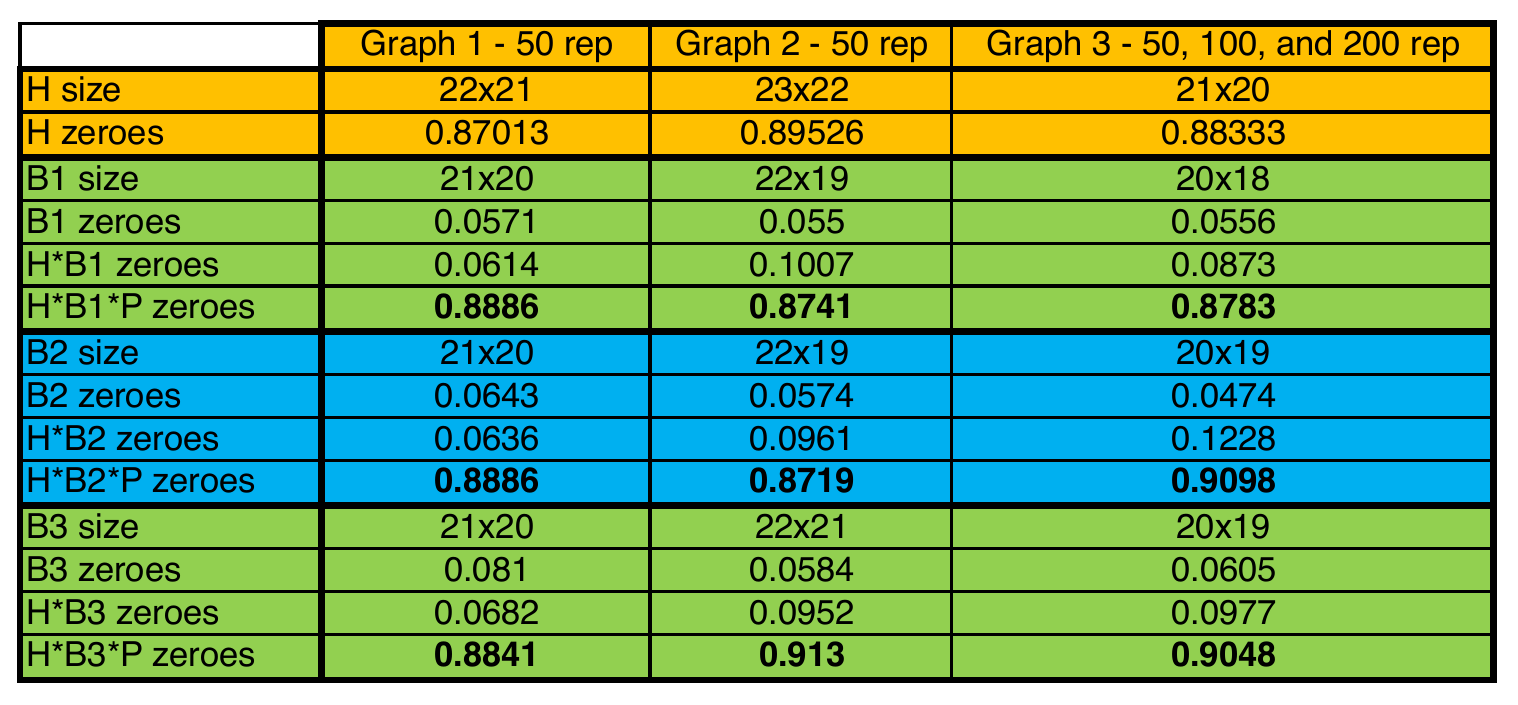}
\caption{Sparsification results for deterministic coding matrices over $\text{GF}(2^4)$.}
\label{fig:sparsification table}
\end{figure}
}
Figure \ref{fig:rate01} includes additional simulation results, for decoding of random LDPC's achieved by matrix sparsification algorithm - Algorithm \ref{alg mat_spar}. We can clearly see that better sparsification yields better bit error rate. Since for the same number of repetitions, Algorithm \ref{alg mat_spar} gives better sparsification for smaller matrices (for larger matrix, more repetitions are needed to achieve better solution for the min-unsatisfy problem), the best BER is achieved by the smallest ($128\times 115$) matrix.

Figure \ref{fig:gauss} shows the density (of ones) achieved by the Algorithm \ref{alg mat_spar}, compared to Gauss elimination and the rate-distortion function $D(R)$, which is the lower bound (as we have shown in Lemma \ref{lem lower bound}) for any matrix  sparsification algorithm. We can see that a single repetition of min-unsatisfy in the Algorithm \ref{alg mat_spar} results in almost the same sparsification achieved by the Gauss elimination. As the number of repetition grows, the density becomes much closer to the lower bound $D(R)$. The matrices used in the simulation are: $300\times 240$, and $300\times 270$ which correspond to the rates $0.8$ and $0.9$ respectively. 
% rate 01 figure
\def\FIGF{
\begin{figure}
\centering
\includegraphics[scale=0.55]{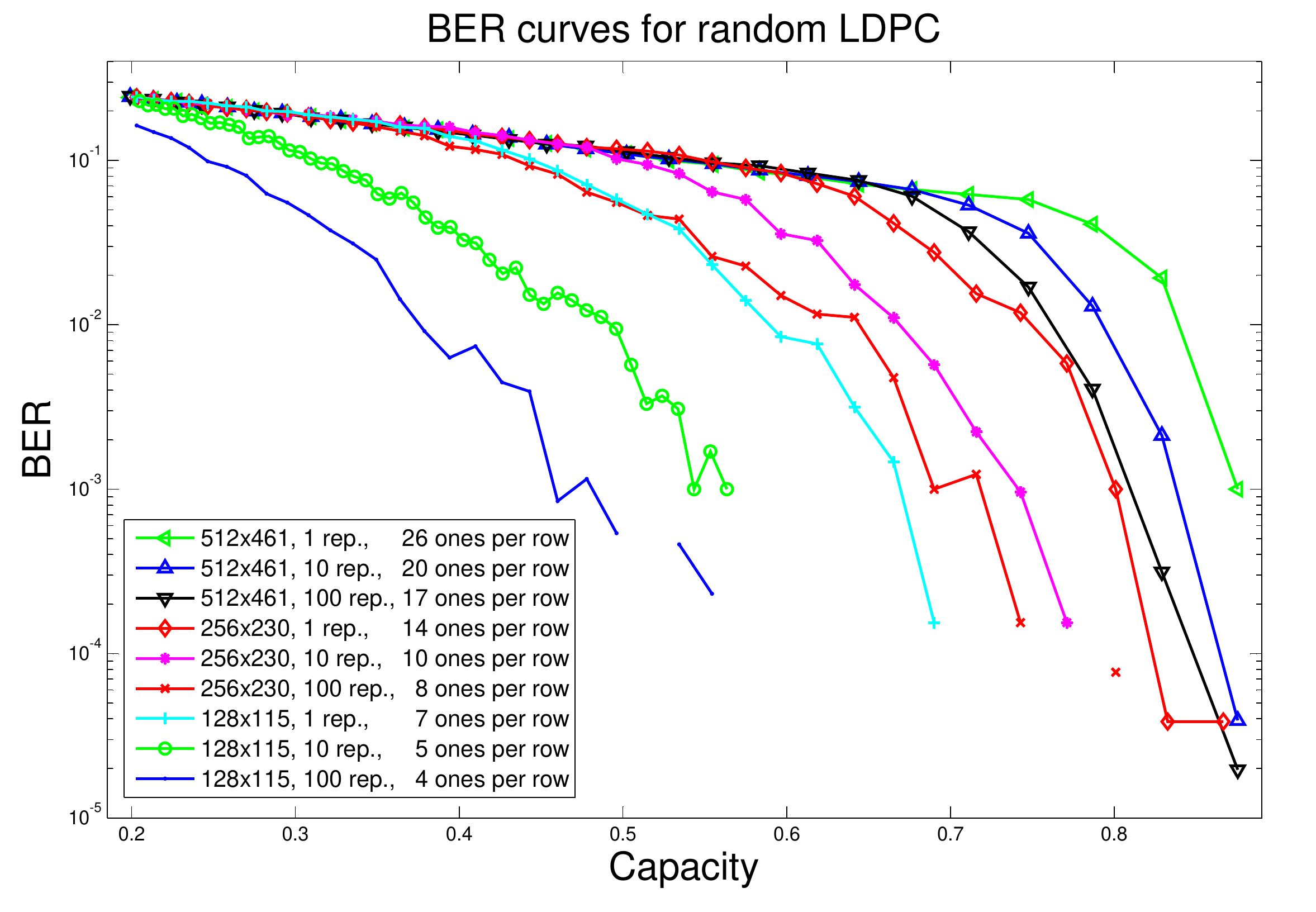}
\caption{Bit error rates for LDPCs created by applying Algorithm \ref{alg mat_spar} on binary symmetric i.i.d.\ matrices of size $n \times nR$, $R=0.9$ (a rate $0.1$ channel code).}
\label{fig:rate01}
\end{figure}
}
% Gauss elimination
\def\FIGG{
\begin{figure}
\centering
\includegraphics[scale=0.55]{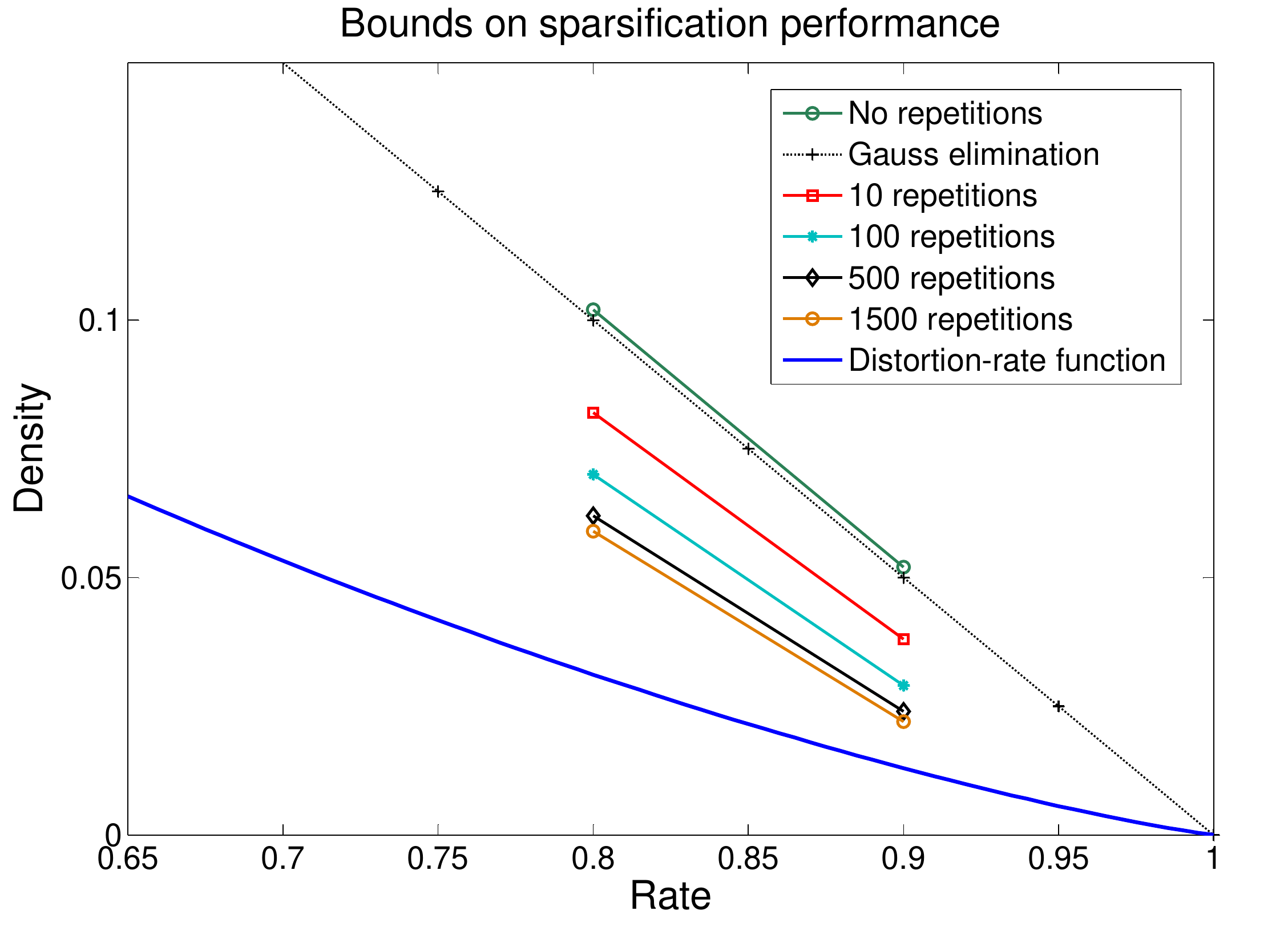}
\caption{Approaching the minimum possible density using Algorithm \ref{alg mat_spar}.}
%Optimal matrix sparsification, which achieved the distortion-rate curve, compared to Gauss Elimination, which achieves a linear curve.}
\label{fig:gauss}
\end{figure}
}

On the more practical side, we note that further sparsification can be sought using methods to sparsify sparse matrices \cite{chang1992hierarchical}. Moreover, the benefit in sparse network coding matrices goes beyond the problem of coding with side information. For example, more efficient algorithms for solving linear equations (in the simple multicast scenario) can be used \cite{wiedemann86sparse_equations,gol1992portrait}. 
%%%%%%%%%%%%%%%%%%%%%%%%%%%%%%%
\section{Conclusion}\label{conc}
In this work, we formally defined joint network-source coding for networks where multiple receivers have side information. We described a code-design procedure which is based on matrix sparsification, enabling each terminal in the network, at the limit of high rates, to receive codewords corresponding to low-density parity-check codes, thus facilitating efficient decoding. Since our scheme is not based on network-source coding separation, optimal rates, matching the cut-set bounds, can be achieved. Simulations performed depict encouraging results, also at non-limiting rates.  
% Comment - random binning.
% that's all folks
\begin{singlespace}
\bibliographystyle{IEEEbib}
%\bibliography{all_references}

\end{singlespace}
\newpage
\FIGA
\FIGB
\FIGHandB
\FIGC
\FIGD
\FIGE
\FIGF
\FIGG

\clearpage
\newpage
\appendix
%%%%%%%%%%%%%%%%%%%%%%%%%%%%%%%%%%%%%%%
\section{Proof of Proposition \ref{prop. can get enough equations}}\label{app prop can get enough equations}
Proposition \ref{prop. can get enough equations} merely translates Corollary \ref{cor. broadcast code}, which asserts that linear network codes such that each node $t$ can receive $\mathrm{maxflow}(s,t)$ independent equations on the $w$ source symbols \emph{over some base field $\text{GF}(2^m)$}, to a \emph{binary} representation. Hence, each node $t$ can receive $m\cdot \mathrm{maxflow}(s,t)$ independent equations on the $mw$ bits. The proof is based on the following technical claim, which shows that indeed linear equations on the elements of the source vector over $\text{GF}(2^m)$ translate to linear equations on the bits in the bit representation of the symbols. Proposition \ref{prop. can get enough equations} will then easily follow from Corollary \ref{cor. broadcast code}.
%%%
\begin{claim}\label{claim linear}
For any set of linear equations over $F=\text{GF}(2^m)$, $\bs M = \by$, where $\bs \in F^w$, $M \in F^{w \times l}$ with rank $r$ and $\by \in F^l$, there exists a set of linear equations over $\text{GF}(2)$, and a binary matrix $B \in \text{GF}(2)^{wm \times lm}$ of rank $rm$, such the binary representations of $\bs$ and $\by$, denoted $\bb_s$ and $\bb_y$ (that is, the coefficients over $\text{GF}(2)$ in the representation of the elements of $F$ as polynomials), satisfies $\bb_s B= \bb_y$ over $\text{GF}(2)$.
\end{claim} 
%%%
\begin{proof}
The system $\bs M = \by$ represents $r$ independent equations over $F$. Each equation is of the form
\begin{equation}\label{one equation}
s_1 m_{1,i} + s_2 m_{2,i} + \ldots + s_w m_{w,i} = y_i
\end{equation}
(with a slight abuse of notation, $m$ represents the extension degree and $m_{i,j}$ represents the $i,j$ entry of $M$).
Each unknown $s_j$ can be represented as a polynomial in $x$,
\[
\alpha_j^0 + \alpha_j^1 x + \ldots + \alpha_j^{m-1} x^{m-1},
\]
where now $\{\alpha_j^l\}_{l=0}^{m-1}$ are $m$ unknowns over $\text{GF}(2)$. The same holds for all entries of the matrix $M$, yet, since $M$ is known, so do the coefficients of each $m_{i,j}$ in its polynomial representation. Hence, \eqref{one equation} translates to
\begin{multline}\label{one equation poly}
(\alpha_1^0 + \alpha_1^1 x + \ldots + \alpha_1^{m-1} x^{m-1})\cdot
(m_{1,i}^0 + m_{1,i}^1 x + \ldots + m_{1,i}^{m-1} x^{m-1})+\ldots
\\
 = (y_i^0 + y_i^1 x + \ldots + y_i^{m-1}x^{m-1}).
\end{multline} 
The l.h.s.\ of \eqref{one equation poly} has a maximal degree $2m-2$, hence, the reminder modulo an irreducible polynomial of degree $m-1$ is computed. This operation is linear in the coefficients of each $x^i$, and, as a result, the l.h.s.\ translates to a polynomial of degree at most $m-1$ with coefficients which are linear functions of the unknowns $\{\alpha_j^l\}_{j=1,l=0}^{w,m-1}$. When comparing the coefficients of each power of $x$, the result is $m$ linear equations for the unknowns $\{\alpha_j^l\}_{j=1,l=0}^{w,m-1}$. Clearly, $r$ independent linear equations of the form of \eqref{one equation} (over $\text{GF}(2^m)$) will result in $mr$ linear equations for $\{\alpha_j^l\}_{j=1,l=0}^{w,m-1}$ (over $\text{GF}(2)$). To see that the resulting $mr$ equations are indeed independent, note that if $w-r$ of the unknowns $s_1, \ldots, s_w$ are given, this method must determine the values of the rest. Since the remaining $r$ unknowns over $\text{GF}(2^m)$ are represented by $mr$ unknowns over $\text{GF}(2)$, the $mr$ equations over $\text{GF}(2)$ must be independent.
\end{proof}
%%%%%%%%%%%%%%%%%%%%%%%%%%%%%%%%%%%%%%%%%%%%%%%%%%%%%%%%%%%%%%%%%%%%%%%%%%%%%%%%%%%%%%
\section{Proof of Proposition \ref{prop properties}}\label{app prop properties}
\begin{proof}[Proof sketch]
\begin{enumerate}\itemsep = -0.2cm
\item If $B_t$ is full rank, non of its columns is the all-zeros column. Assume column $j$ of $B_t$ contains $l$ ones. Then, for each $i$, the entry $(i,j)$ of $HB_t$ is the sum of $l$ independent Bernoulli$(\lambda/n)$ random variables, and the summands for the entry $(i,j)$ are independent of those for $(i',j)$, for any $i \ne i'$. 
\item If $B_t$ is full rank, non of its columns are identical. This means that the sums composing $(HB_t)(i,j)$ and $(HB_t)(i,j')$ differ in at least one random variable. Since this random variable is binary, uniform and independent of all the rest, $(HB_t)(i,j)$ and $(HB_t)(i,j')$ are also uniform and independent. For different output rows this holds trivially.
\item \label{item if linear}The event that $(HB_t)(i,j) = 0$ is the event that an even number of the Bernoulli random variables is $1$, hence
\begin{equation}\label{prob. of even}
\prob\left((HB_t)(i,j) = 0\right) = \frac{1}{2}\left(1+\left(1-\frac{2\lambda}{n}\right)^l\right)
\end{equation} 
for all $i$. 
Assume $l=c_j n$ and compute the limit as $n \to \infty$. For large enough $\lambda$, the distribution is arbitrarily close to uniform bits.
\item \label{if sub linear than 0} We use the result of Item \ref{item if linear} above. Since the number of non-zero entries in column $j$ of $B_t$ is sub-linear in $n$, for large enough $n$ it is smaller than any $cn$ (for arbitrary small $c>0$). Thus, 
\begin{equation}
\frac{1}{2}\left(1+\left(1-\frac{2\lambda}{n}\right)^{cn}\right) \le
\prob\Big((HB_t)(i,j) = 0\Big) = \frac{1}{2}\left(1+\left(1-\frac{2\lambda}{n}\right)^l\right)
\leq 1.
\end{equation}
Taking the limit of $n \to \infty$ we have 
\begin{equation}
\frac{1}{2}\left(1+e^{-2c\lambda}\right) \le
\prob\left((HB_t)(i,j) = 0\right)
\leq 1
\end{equation}
for arbitrarily small $c$, which completes the proof.
\item If $d(\bb_{j},\bb_{j'}) = \Theta(n)$, then the difference between $(HB_t)(i,j)$ and $(HB_t)(i,j')$ is a sum of a linear number of independent Bernoulli$(\lambda/n)$ random variables. This sum is a uniform bit if $\lambda = \omega(n)$, hence the random variables $(HB_t)(i,j)$ and $(HB_t)(i,j')$ are independent for each $i$. However, if this sum is sub-linear, then a constant $\lambda$ results in an asymptotically zero probability for it to be $1$ (similar to item \ref{if sub linear than 0} above). This means that the probability that $(HB_t)(i,j)$ and $(HB_t)(i,j')$ \emph{differ} is small. Note that in this case, with high probability, the \emph{sum} of $(HB_t)(i,j)$ and $(HB_t)(i,j')$ is $0$, hence summing the two columns gives a sparse column.  
\end{enumerate}
\end{proof}
%%%%%%%%%%%%%%%%%%%%%%%%%%%%%%%%%%%%%%%%%%
\end{document}